\documentclass[twocolumn,10pt]{IEEEtran}
\usepackage{booktabs}

\input{def.tex}
\usepackage{dsfont}
\usepackage{minibox}
\DeclareSymbolFont{matha}{OML}{txmi}{m}{it}% txfonts
\DeclareMathSymbol{\varv}{\mathord}{matha}{118}
\usepackage{eurosym}
\usepackage{hhline,physics}

\usepackage{multicol}
\usepackage{algorithm}
\usepackage{graphicx}
\usepackage{textcomp}
\usepackage{caption,pifont}
\usepackage[multiple]{footmisc}

\usepackage{algorithmic}

\makeatletter
\renewcommand{\baselinestretch}{0.89}
\newcommand{\floor}[1]{\lfloor #1 \rfloor}

\IEEEoverridecommandlockouts

\begin{document}
	\title{Achievable Rate Optimization for Large Stacked Intelligent Metasurfaces Based on Statistical CSI} 
	\author{Anastasios Papazafeiropoulos, Pandelis Kourtessis,  Symeon Chatzinotas, Dimitra I. Kaklamani, 			Iakovos S. Venieris \thanks{A. Papazafeiropoulos is with the Communications and Intelligent Systems Research Group, University of Hertfordshire, Hatfield AL10 9AB, U. K., and with SnT at the University of Luxembourg, Luxembourg.  P. Kourtessis is with the Communications and Intelligent Systems Research Group, University of Hertfordshire, Hatfield AL10 9AB, U. K.  S. Chatzinotas is with the SnT at the University of Luxembourg, Luxembourg. Dimitra I. Kaklamani is with the Microwave and Fiber Optics Laboratory, and Iakovos S. Venieris is  with the Intelligent Communications and Broadband Networks Laboratory, School of Electrical and Computer Engineering, National Technical University of Athens, Zografou, 15780 Athens,	Greece.	
			%				A. Papazafeiropoulos was supported  by the University of Hertfordshire's 5-year Vice Chancellor's Research Fellowship. S. Chatzinotas   was supported by the National Research Fund, Luxembourg, under the project RISOTTI.
			Corresponding author's email: tapapazaf@gmail.com.}}
	\maketitle\vspace{-1.7cm}
	\begin{abstract}	
		Stacked intelligent metasurface (SIM) is an emerging design that consists of  multiple layers of metasurfaces. A SIM enables holographic
		multiple-input multiple-output (HMIMO) precoding  in the wave domain, which results in the reduction of energy consumption and hardware cost. On the ground of multiuser beamforming,  this letter  focuses on  the downlink achievable rate and its maximization. Contrary to previous works on multiuser SIM, we consider statistical channel state information (CSI) as opposed to instantaneous CSI to overcome  challenges such as large overhead. Also, we examine the performance of large surfaces. We apply an alternating optimization (AO) algorithm regarding the phases of the SIM and the allocated transmit power. Simulations  illustrate the performance of the considered large SIM-assisted design as well as the comparison between different CSI considerations.
	\end{abstract}
	
	\section{Introduction}
	
	The technology of reconfigurable intelligent surfaces (RISs)	has recently emerged to increase coverage and enhance spectral and energy efficiencies in various communication environments \cite{DiRenzo2020,Liu2021}. In general terms, an RIS includes a surface that includes a large number of elements, which  are nearly passive and have low cost. The purpose of these elements is to adjust the phases of the incident electromagnetic (EM) waves by using a smart controller, and hence, shape the propagation environment dynamically  \cite{Huang2019,Papazafeiropoulos2021,Papazafeiropoulos2023}.
	
	However, most existing works on RIS assume single-layer metasurface structure \cite{Huang2019,Papazafeiropoulos2021,Guo2020a}, which imposes a constraint on the adjustment of the beam patterns. Also, the single-layer structures of RISs do not have  the capability of   inter-user interference suppression as shown in \cite{Guo2020a}. These observations led the authors in \cite{An2023d,An2023c} to propose a stacked intelligent metasurface (SIM), which consists of an array of programmable metasurfaces similar to artificial neural networks (ANNs). Among the processing capabilities of a SIM, we   that the \textcolor{black}{forward}  propagation takes place at the speed of light. 
	
	On this ground, in \cite{An2023d}, authors proposed a SIM-based  design for the transceiver of point-to-point multiple-input multiple-output (MIMO) communication systems, where the combining and the precoding  take place as the EM waves propagate along the SIM. In \cite{An2023c},  we observe the integration of a SIM to the transmitter, i.e., the  base station (BS) towards enabling  beamforming  in the EM domain based on instantaneous channel state information (CSI). Contrary to \cite{An2023d} and \cite{An2023c}, in \cite{Papazafeiropoulos2024a} and \cite{Papazafeiropoulos2024}, we proposed more general hybrid digital wave designs, where all element parameters are optimised simultaneously through more efficient algorithms. 
	
	In this work, we focus on a SIM-enabled multiuser architecture operating solely in the wave domain. Note that \cite{Papazafeiropoulos2024} assumes a hybrid digital wave design, \textcolor{black}{and \cite{Lin2024} focuses on satellite communication systems}. Also, contrary to previous works \cite{An2023d,An2023c,Papazafeiropoulos2024a}, we consider a SIM that consists of large metasurfaces, since we apply the  use-and-then-forget (UatF) bound \cite{Bjoernson2017}. Most importantly, we obtain the downlink rate and perform its optimization regarding the phase shifts and transmit power in terms of statistical CSI. Notably, this approach enables the optimization at every several coherence intervals  rather than optimizing at each interval. Hence, we achieve lower overhead,  which is one of the main challenges in SIM-assisted systems.
	
	\textcolor{black}{	\textit{Notation}: Matrices  and  vectors are represented by boldface upper  and lower case symbols, respectively. The notations $(\cdot)^\T$, $(\cdot)^\H$, and $\tr\!\left( {\cdot} \right)$ denote the transpose, Hermitian transpose, and trace operators, respectively. Also, the symbol  $\EE\left[\cdot\right]$ denotes  the expectation operator. The floor function $\floor{x}$  gives as output the greatest integer less than or equal to $x$. The notation  $\diag\left(\bA\right) $ represents a vector with elements equal to the  diagonal elements of $ \bA $. The notation  $\bb \sim \cC\cN{(\b0,\mathbf{\Sigma})}$ represents a circularly symmetric complex Gaussian vector with zero mean and a  covariance matrix $\mathbf{\Sigma}$.}
	
	\section{System Model}\label{System}
	We consider a SIM-aided MIMO communication system as depicted in Fig. \ref{Fig01}. In particular, a BS, which includes  $ N_{t} $ antennas, communicates with $ K $ single-antenna user equipments (UEs) through a  SIM  performing wave-based processing. The SIM is implemented by  $ L $ metasurfaces, where each one has a large number of $ N $ meta-atoms.   Let  $\mathcal{K}=\{1,\ldots,K\} $, $ \mathcal{L}=\{1,\ldots,L\} $, and $ \mathcal{N}=\{1,\ldots,N\} $ denote the sets of UEs, metasurfaces, and  meta-atoms,  respectively. Note that an  intelligent controller adjusts the shifts of the phases of the electromagnetic (EM) waves that impinge on  the metasurface layers. 
	%In other words, we   rely on the SIM multiuser design proposed in \cite{An2023}.

	\begin{figure}
		\begin{center}
			\includegraphics[width=0.8\linewidth]{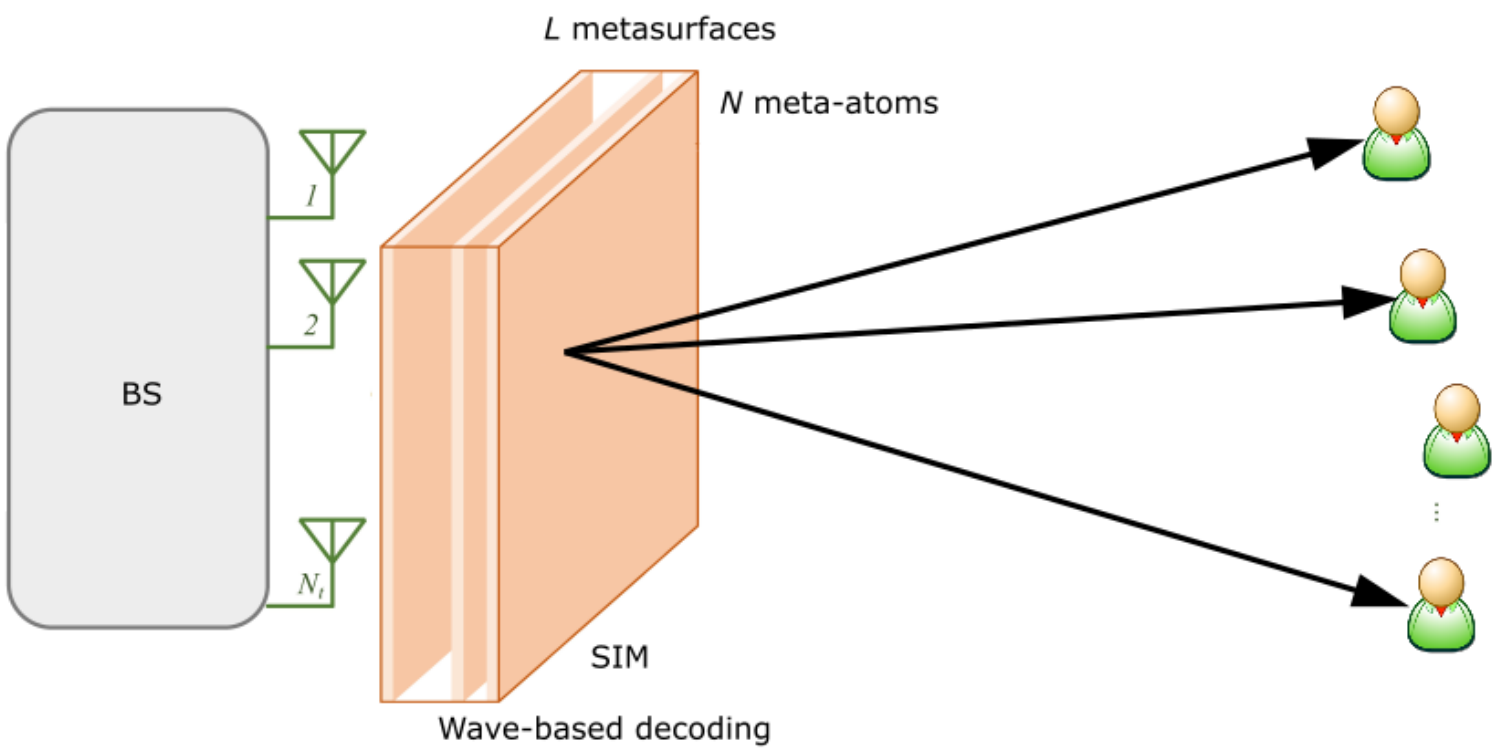}
			\caption{A SIM-aided MIMO system. }
			\label{Fig01}
		\end{center}
	\end{figure}

	On this basis, let $ \theta_{n}^{l}\in [0,2\pi), n \in \mathcal{N}, l \in \mathcal{L} $ be the phase shift by  the $ n $th meta-atom   on the   surface layer $ l $. Also, we denote  $ \phi_{n}^{l} =e^{j \theta_{n}^{l}}$, and $ \bPhi_{l}=\diag(\bphi^{l})\in \mathbb{C}^{N \times N} $, where $ \bphi^{l} =[\phi^{l}_{1}, \dots, \phi^{l}_{N}]^{\T}\in \mathbb{C}^{\textcolor{black}{N} \times 1}$.\footnote{Herein, we consider  phase shifts, which are continuously-adjustable and their modulus equals to $ 1 $ to evaluate  large SIM-aided MIMO communications. Practical issues  such as the consideration of discrete phase shifts \cite{Liu2022} is the topic of future work.} In addition, $ \bW^{l}\in \mathbb{C}^{N \times N}, l \in \mathcal{L}/\{1\} $ denotes the  coefficient matrix between     layer $ (l-1)$ and     layer $ l $. In particular, its entries from  meta-atom  $ \tilde{n} $ on  layer $ (l-1) $ to  meta-atom  $ n $ on  layer $ l, \forall l \in \mathcal{L} $  are given by 
	\begin{align}
		w_{n,\tilde{n}}^{l}=\frac{A_{t}cos x_{n,\tilde{n}}^{l}}{r_{n,\tilde{n}}^{l}}\left(\frac{1}{2\pi r^{l}_{n,\tilde{n}}}-j\frac{1}{\lambda}\right)e^{j 2 \pi r_{n,\tilde{n}}^{l}/\lambda},\label{deviationTransmitter}
	\end{align}
	where $ A_{t} $ is the area of each  meta-atom  at the SIM, $ x_{n,\tilde{n}}^{l} $ denotes the angle between the normal direction of the  transmit metasurface layer  $ (l-1) $ and the propagation direction,  $ r_{n,\tilde{n}}^{l} $, is the respective transmission distance. Moreover, let $ \bw^{1}_{k} \in \mathbb{C}^{N \times  1} $  express the  coefficient from the transmit antenna array. Thus, the impact of the  SIM can be expressed as
	\begin{align}
		\bG=\bPhi_{L}\bW^{L}\cdots\bPhi_{2}\bW^{2}\bPhi_{1}\in \mathbb{C}^{N \times N}.\label{TransmitterSIM}
	\end{align}

	Let $ \bh_{k}\in \mathbb{C}^{N \times 1} , \forall k \in \mathcal{K}$ express the  channel between the last  layer and UE $ k $ that is described by  the correlated Rician fading distribution as
	\begin{align}
		\bh_{k}=\sqrt{\beta_{k}}\left(\sqrt{\frac{\kappa_{k}}{1+\kappa_{k}}}\bh_{k,\mathrm{LoS}}+\sqrt{\frac{1}{1+\kappa_{k}}}	\bh_{k,\mathrm{NLoS}}\right)~~\forall k\in \mathcal{K}.\label{channel}
	\end{align}
	In \eqref{channel}, $ \kappa_{k} $ is the Rician factor, $ \beta_{k} $ is the channel gain,	$ \bh_{k,\mathrm{NLoS}} \in \mathbb{C}^{N \times 1}$ is the LoS component, and $ \bh_{k,\mathrm{NLoS}}\sim \mathcal{CN}(\b0,\bR) \in \mathbb{C}^{N \times 1}$ is the NLoS component  with $ \bR\in \mathbb{C}^{N\times N}  $  representing the spatial correlation of each surface. This correlation  is obtained $\forall n\in \mathcal{N}, \tilde{n}\in \mathcal{N}$ as \cite{Bjoernson2020}
	\textcolor{black}{	\begin{align}
			[\bR_{\mathrm{SIM}}]_{\tilde{n},n}&=\mathrm{sinc}(2 \|\bu_{n}-\bu_{\tilde{n}}\|/\lambda), n,\tilde{n}=1,\ldots,N\label{corC}
		\end{align}
		where  $\bu_{n}=[0, i(n)d_{\mathrm{H}}, j(n)d_{\mathrm{V}}]^{\T}$ with $i(n)=\mod(n-1,N_{x})$ and $j(n)=\floor{(n-1)/N_{x}}$ being the horizontal and vertical indices of element $n$, respectively. $N_{x}$ and $N_{y}$ are the elements per row and column, while $d_{\mathrm{H}} $ and $d_{\mathrm{V}}$ denote the horizontal width and the vertical height.}
	
	\section{Downlink Data Transmission}\label{PerformanceAnalysis}
	During the downlink transmission and based on wave-based beamforming \textcolor{black}{\cite{An2023c}}, the received signal at the $ k $-th UE is written as
	\begin{align}
		y_{k}=\bh_{k}^{\H}\bG \sum_{i=1}^{K}\bw_{i}^{1}\sqrt{p}_{i}x_{i}+n_{k},~~~\forall k \in \mathcal{K}
	\end{align}
	where $ x_{i} $ is
	the  information symbol intended for the $ k $-th UE, which has
	a zero mean and unit variance. Also, $ p_i $ is the power corresponding to the $ k $-th UE 	with $ \sum_{i=1}^{K}p_{i} \le P_{\mathrm{T}}$, where  $ P_{\mathrm{T}} $ is  the total transmit power  at the BS. Also $ n_{k}\sim \mathcal{CN}(0,\sigma_{k}^{2}) $ denotes  the additive white Gaussian noise (AWGN) with $ \sigma_{k}^{2} $ expressing its variance at UE $ k $.

	The downlink achievable SE of UE $ k $ is given by
	\begin{align}
		\mathrm{SE}	=\sum_{k=1}^{K}\log_{2}\left ( 1+\gamma_{k}\right)\!,\label{LowerBound}
	\end{align}
	where  $ \gamma_{k}$ denotes the downlink  signal-to-interference-plus-noise ratio (SINR), which is written according to the UaTF bounding technique \cite{Bjoernson2017} as
	\begin{align}
		\gamma_{k}=	\frac{p_{k}	|\EE\{\bh_{k}^{\H}\bG \bw_{k}^{1}\}|^{2}}{\sum_{i=1}^{K}p_{i}\EE\{|\bh_{k}^{\H}\bG \bw_{i}^{1}|^{2}\}-	p_{k}|\EE\{\bh_{k}^{\H}\bG \bw_{k}^{1}\}|^{2}+\sigma_{k}^{2}},\label{gamma1}
	\end{align}
	where  is assumed  that UE $ k $ has knowledge of the average  effective channel.
	
	\begin{proposition}\label{proposition:DLSINR}
		The  achievable SINR of UE $k$ for a given SIM during the downlink transmission  is provided by \eqref{gamma2}.
		\begin{figure*}
			\begin{align}
				\gamma_{k}=	\frac{p_{k}	\kappa_{k}|\bh_{k,\mathrm{LoS}}^{\H}\bG \bw_{k}^{1}|^{2}}{\sum_{i=1}^{K}p_{i}\tr(\bG \bw_{i}^{1}\bw_{i}^{1^{\H}}\bG^{\H}\bR)+\sum_{i\ne k}^{K}p_{i}	\kappa_{k}\bh_{k,\mathrm{LoS}}^{\H}\bG \bw_{i}^{1}\bw_{i}^{1^{\H}}\bG^{\H}\bh_{k,\mathrm{LoS}}+\frac{\sigma_{k}^{2}(1+\kappa_{k})}{\beta_{k}}}\label{gamma2}.
			\end{align}
			\hrulefill
		\end{figure*}
	\end{proposition} 
	\begin{proof}
		The numerator becomes
		\begin{align}
			|\EE\{\bh_{k}^{\H}\bG \bw_{k}^{1}\}|^{2}=\beta_{k}\frac{\kappa_{k}}{1+\kappa_{k}}|\bh_{k,\mathrm{LoS}}^{\H}\bG \bw_{k}^{1}|^{2}.\label{num1}
		\end{align}
		Regarding the denominator of \eqref{gamma1}, the first term is written as
		\begin{align}
			&	\EE\{|\bh_{k}^{\H}\bG \bw_{i}^{1}|^{2}\}=\tr(\bh_{k}^{\H}\bG \bw_{i}^{1}\bw_{i}^{1^{\H}}\bG^{\H}\bh_{k})\label{denom1}\\
			&=\beta_{k}\frac{1}{1+\kappa_{k}}\tr(\bG \bw_{i}^{1}\bw_{i}^{1^{\H}}\bG^{\H}\bR)\nn\\
			&	+\beta_{k}\frac{\kappa_{k}}{1+\kappa_{k}}\bh_{k,\mathrm{LoS}}^{\H}\bG \bw_{i}^{1}\bw_{i}^{1^{\H}}\bG^{\H}\bh_{k,\mathrm{LoS}},\label{denom2}
		\end{align}
		where, in \eqref{denom1}, we have applied that  $\bx^{\H}\by = \tr(\by \bx^{\H})$ for any vectors $\bx$, $\by$.
		By substituting \eqref{num1} and \eqref{denom2} in \eqref{gamma1}, we obtain the achievable SINR.
	\end{proof}
	
	\section{Problem Formulation  and  Optimization}\label{PSConfig}
	The  maximization of the sum SE regarding the phase shifts  of each surface and the allocated power is of great importance.
	
	\subsection{Problem Formulation}
	The maximization problem is formulated as
	\begin{subequations}\label{eq:subeqns}
		\begin{align}
			(\mathcal{P})~~&\max_{\bphi_{l},\bp} 	\;	f(\bphi_{l},\bp)=\sum_{k=1}^{K}\log_2\left(1+\frac{D_k}{I_k}\right)\label{Maximization1} \\
			&~	\mathrm{s.t}~~~	\bG=\bPhi_{L}\bW^{L}\cdots\bPhi_{2}\bW^{2}\bPhi_{1},
			\label{Maximization4} \\
			&\;\quad\;\;\;\;\;\!\!~\!		\bPhi_{l}=\diag(\phi^{l}_{1}, \dots, \phi^{l}_{N}), l \in \mathcal{L},
			\label{Maximization5} \\
			&\;\quad\;\;\;\;\;\!\!~\!		|	\phi^{l}_{n}|=1, n \in \mathcal{N}, l \in \mathcal{L},	\label{Maximization7}\\
			&\;\quad\;\;\;\;\;\!\!~\!	\sum_{i=1}^{K}p_{i} = P_{\mathrm{T}}\label{Maximization8}\\
			&\;\quad\;\;\;\;\;\!\!~\!	p_{k} \ge 0, \forall k \in \mathcal{K},
		\end{align}
	\end{subequations}
	where $ D_k $ and $ I_k $ are the numerator and denominator of $ \gamma_{k} $ obtained in Proposition \ref{proposition:DLSINR}. Also, we have defined the vector $ \bp=[p_{1}, \ldots,p_{K}]^{\T}$. Note that
	the constraint \eqref{Maximization7} expresses that each RIS element provides only a phase shift while   \eqref{Maximization8} corresponds to the maximum power constraint.

	The non-convexity optimization problem $ 	(\mathcal{P}) $ and its dependence on the  unit-modulus constraint with respect to $ 	\bphi_{l} $ make the solution challenging. For this reason, we resort to alternating optimization (AO). According to this technique,  $ \bphi_{l} $
	and  $ \bp $ will be optimized individually in an  iterative manner. Specifically, first, we  find  the optimum $ \bphi_{l} $ for  a fixed $ \bp $. During the next step, we solve for $ \bp $ with $ \bphi_{l} $ fixed. The objective converges to  its optimum value  by iterating this process, which leads to the increase of $f(\bphi_{l},\bp) $ after each step  until a specific point because of the upper-bound coming from  the power constraint \eqref{Maximization8}.

	\subsection{SIM Optimization}
	Until now,  $ 	\bphi_{l}$ was assumed fixed.  However, to exploit each metasurface towards wave-based beamforming  while   maximizing  \eqref{LowerBound}, the optimization of each  $ 	\bphi_{l} $ has to take place. Its presence is observed  inside the matrix $ \bG $, appearing in $ D_k $ and $ I_k $. Hence, the  maximization problem regarding $ 	\bphi_{l}  $ is described as
	\begin{subequations}\label{eq:subeqns1}
		\begin{align}
			(\mathcal{P}1)~~&\max_{\bphi_{l}} 	\;	f(\bphi_{l})\label{Maximization11} \\
			&~	\mathrm{s.t}~~~	\bG=\bPhi_{L}\bW^{L}\cdots\bPhi_{2}\bW^{2}\bPhi_{1},
			\label{Maximization41} \\
			&\;\quad\;\;\;\;\;\!\!~\!		\bPhi_{l}=\diag(\phi^{l}_{1}, \dots, \phi^{l}_{N}), l \in \mathcal{L},
			\label{Maximization51} \\
			&\;\quad\;\;\;\;\;\!\!~\!		|	\phi^{l}_{n}|=1, n \in \mathcal{N}, l \in \mathcal{L},	\label{Maximization71}
		\end{align}
	\end{subequations}
	where  the maximization problem  $ 	(\mathcal{P}1) $ is  non-convex regarding  $ 	\bphi_{l} $, and it obeys to  a unit-modulus constraint with respect to $ 	\phi^{l}_{n} $. Application of the projected gradient 	ascent algorithm until convergence while taking into account the unit-modulus constraint  results in a locally optimal solution to $ 	(\mathcal{P}1) $.
	
	%Let the set
	%\begin{align}
	%	\Phi_{l}&=\{\bphi_{l}\in \mathbb{C}^{N  \times 1}: |\phi^{l}_{i}|=1, i=1,\ldots, N\}.
	%\end{align}
	The proposed algorithm suggests starting from $ \bphi_{l}^{0} $, and then shifting along the gradient of   $  f(\bphi_{l}) $. The new point $ \bphi_{l} $ is  projected onto  $ \Phi_{l} $ to hold the new points in the feasible set. Specifically, the unit-modulus constraint  means that $ \phi^{l}_{n} $ has to be found inside the unit circle. $ P_{\Phi_{l}}(\cdot) $ is the projection onto $ \Phi_{l} $. Hence, we have
	\begin{align}
		\bar{u}_{l,n}=\left\{
		\begin{array}{ll}
			\frac{u_{l,n}}{|u_{l,n}|} & u_{l,n}\ne 0 \\
			e^{j \phi^{l}_{n}}, \phi^{l}_{n} \in [0, 2 \pi] &u_{l,n}=0 \\
		\end{array}, n=1, \ldots,N,
		\right.
	\end{align}
	where   the  vector $ \bar{\bu}_{l} $ of $ P_{\Phi_{l}}(\bu_{l}) $ is a given point.
	
	The algorithm is described by the following iteration
	\begin{align}
		\bphi_{l}^{i+1}&=P_{\Phi_{l}}(\bphi_{l}^{i}+\mu_{i}\nabla_{\bphi_{l}}f(\bphi_{l}^{i}))\label{p1}.
	\end{align}
	The Armijo-Goldstein backtracking line search method provides the step size, which is $ \mu_{i} = L_{i}\kappa^{m_{i}} $, where   $ \kappa \in (0,1) $ and $ L_i>0 $. Note that  $ m_{i} $ is the
	smallest positive integer that  satisfies
	\begin{align}
		f(\bphi_{l}^{i+1})\geq	Q_{L_{i}\kappa^{m_{i}}}(\bphi_{l}^{i};\bphi_{l}^{i+1}),
	\end{align}
	where 
	\begin{align}
		\!\!	Q_{\mu}(\bphi_{l};\bx)\!=\!f(\bphi_{l})\!+\!\langle	\nabla_{\bphi_{l}}f(\bphi_{l}),\bx\!-\!\bphi_{l}\rangle\!-\!\frac{1}{\mu}\|\bx-\bphi_{l}\|^{2}_{2}
	\end{align}
	is the  quadratic approximation of $ f(\bphi_{l}) $.

	\begin{proposition}\label{propositionGradient}
		The gradient of $f(\bphi_{l}) $ regarding  $\bphi_{l}^{*}$ is obtained in	closed-form as	
		\begin{align}
			\nabla_{\bphi_{l}}f(\bphi_{l})&=\frac{1}{\log_{2}(e)}\sum_{k=1}^{K}\frac{{I}_{k}\nabla_{\bphi_{l}}D_k-D_k\nabla_{\bphi_{l}}{I}_{k}}{(1+\gamma_{k}){I}_{k}},\label{grad11}	
		\end{align}
		where
		\begin{align}
			&\nabla_{\bphi_{l}}D_k 
			=p_{k}	\kappa_{k}\diag(\bC_{l}^{*}\bw_{k}^{1^{*}}\bw_{k}^{1^{\T}}\bG^{\T}	\bh_{k,\mathrm{LoS}}^{*}\bh_{k,\mathrm{LoS}}^{\T}\bA_{l}^{*})\label{differentialPhi71},\\
			&\nabla_{\bphi_{l}}{I}_{k}
			=\sum_{i=1}^{K}p_{i}\diag(\bC_{l}^{*}\bw_{i}^{1^{*}}\bw_{i}^{1^{\T}}\bG^{\T}	\bR\bA_{l}^{*})\nn\\
			&+\sum_{i\ne k}^{K}p_{i}	\kappa_{k}\diag(\bC_{l}^{*}\bw_{i}^{1^{*}}\bw_{i}^{1^{\T}}\bG^{\T}	\bh_{k,\mathrm{LoS}}^{*}\bh_{k,\mathrm{LoS}}^{\T}\bA_{l}^{*})\label{gradI10}
		\end{align}
		with  $ 	\bA_{l}=\bPhi_{L}\bW^{L}\cdots\bPhi_{l+1}\bW^{l+1} $, and $\bC_{l}= \bW^{l}\bPhi_{l-1}\bW^{l-1}\cdots \bPhi_{1} $.
	\end{proposition}
	\begin{proof}
		Please see Appendix~\ref{proposition1}.	
	\end{proof}
	
	\textcolor{black}{The SIM optimization design, based on the gradient ascent, appears a significant advantage  because  the gradient ascent is obtained in a closed form. It has low computational complexity because it consists of simple matrix operations. Specifically,  the complexity of \eqref{Maximization11} for large SIMs is $ \mathcal{O}(N_{t}N^2+LN^{2} +KN^{3} )$, and the complexity of \eqref{grad11} is similar. In other words,  the number of  meta-atoms of each surface has a higher  impact.}
	\subsection{Power Optimization}
	Given a fixed $	\bPhi_{l}$, we focus on the optimization with respect to $ \bp $. Specifically, we have 
	\begin{subequations}\label{eq:subeqns2}
		\begin{align}
			(\mathcal{P}2)~~&\max_{\bp} 	\;	f(\bp)\label{Maximization3} \\
			&\;\quad\;\;\;\;\;\!\!~\!	\sum_{i=1}^{K}p_{i} = P_{\mathrm{T}}, ~	p_{k} \ge 0, \forall k \in \mathcal{K}.
		\end{align}
	\end{subequations}
	
	The nonconvexity of problem $	(\mathcal{P}2)$ leads us to obtain   a solution which is locally optimal. For this reason, we  apply a weighted minimum mean square error (MMSE) reformulation of the sum SE.  To this end, we denote $\bc\!=\![c_{1}, \ldots,c_{K} ]^{\T}  $. Then,  the SINR $\gamma_{k} $ can be written in terms of  the vector $ \bp $ as
	\begin{align}
		\gamma_{k} =\frac{p_{k} q_{k}}{\bc^{\T}\bp+u_{k}^{2}},\label{SINR10}
	\end{align}
	where
	\begin{align}
		&q_{k}=\kappa_{k}|\bh_{k,\mathrm{LoS}}^{\H}\bG \bw_{k}^{1}|^{2}, c_{k}=\beta_{k}\frac{1}{1+\kappa_{k}}\tr(\bG \bw_{k}^{1}\bw_{k}^{1^{\H}}\bG^{\H}\bR),\nn\\
		&t_{k}^{2}=\frac{\sigma_{k}^{2}(1+\kappa_{k})}{\beta_{k}}, c_{i}=\beta_{k}\frac{1}{1+\kappa_{k}}\tr(\bG \bw_{i}^{1}\bw_{i}^{1^{\H}}\bG^{\H}\bR)\nn\\
		&	+\beta_{k}\frac{\kappa_{k}}{1+\kappa_{k}}\bh_{k,\mathrm{LoS}}^{\H}\bG \bw_{i}^{1}\bw_{i}^{1^{\H}}\bG^{\H}\bh_{k,\mathrm{LoS}},~\forall i\ne k.
	\end{align}
	
	Now, we consider the  single‐input and single‐output (SISO) channel model that comes from the SINR in \eqref{SINR10}, which is given by
	\begin{align}
		\tilde{y}_{k}=\sqrt{p_{k} q_{k}}s_{k}+\sum_{i=1}^{K}\sqrt{p_{i}c_{i}}s_{i}+n_{k},
	\end{align}
	where $ n_{k}\sim \cC\cN\left( 0, u_{k}^{2} \right) $ while $ s_{i }\in \mathbb{C} $ is the data signal with unit variance, and $ \tilde{y}_{k} $ is the received signal.
	
	Then, the  receiver  estimates $s_{k}$, i.e., $\hat{s}_{k}=v_{k}^{*} 	\tilde{y}_{k} $ with   $ v_{k} $ being a receiver coefficient. The corresponding  mean square error 	$ e_{k}(\bp,v_{k}) =[|\hat{s}_{k}-s_{k}|^{2}]$ becomes
	\begin{align}
		e_{k}(\bp,v_{k})=v_{kg}^{2}\left(p_{k}q_{k}+\bc_{k}^{\T}\bp+u_{k}^{2}\right)-2 v_{k}\sqrt{q_{k}p_k}+1.\label{mse1}
	\end{align}
	
	For a given $ \bp $, $ 	v_{k} $ is obtained by the  minimization of the MSE as 
	\begin{align}
		v_{k}=\frac{\sqrt{p_{k} q_{k}}}{p_{k}q_{k}+\sum_{i=1}^{K}p_{i}c_{i}+u_{k}^{2}}.
	\end{align}
	
	Inserting $ v_{k} $ into \eqref{mse1}, $ e_{k} $ becomes $ 1/\left(1+\gamma_{k} \right) $. Based on  the  weighted MMSE method, let the auxiliary weight $ d_{k}\ge 0 $ for the MSE $  e_{k}$ and consider the problem
	\begin{align}\begin{split}
			(\mathcal{P}2.1)\min_{\substack{\bp\ge 0,\\\{v_{k}, d_{k}\ge 0: k=1,\ldots,K\}}} 	&		{K}\sum_{k=1}^{K}d_{k} e_{k}(\bp,\bv_{k})-\ln(d_{k})\\
			\mathrm{s.t}~~\;\!&\sum_{i=1}^{K}p_{i}\le P_{\mathrm{T}}.
		\end{split}\label{Maximization10} 
	\end{align}
	
	Problems $ 	(\mathcal{P}2) $ and $ 	(\mathcal{P}2.1) $ are equivalent, and thus, are subject to the same  solution.  The solution of $ 	(\mathcal{P}2.1) $ can be provided in closed form as
	\begin{align}
		p_{i}=\min\left(P_{\mathrm{T}},\frac{q_{k}d_{k}^{2}v_{k}^{2}}{\left(q_{k}d_{k}v_{k}^{2}+\sum_{i=1}^{K}d_{i}v_{i}^{2}c_{i}\right)^{2}}\right).
	\end{align}

	\textcolor{black}{The power allocation presents a similar complexity to the SIM optimization design since it consists of similar  matrix operations to $	(\mathcal{P}1)$, i.e., its complexity is  $ \mathcal{O}(N_{t}N^2+LN^{2} +KN^{3} )$.}

	\begin{remark}
		Both algorithms, corresponding to Problems  	$(\mathcal{P}1)$ and $(\mathcal{P}2)$, have low computation complexity and converge quickly. Note that  the achievement of  a local optimum, obtained from this optimization,  will make different initializations result in different solutions, which will be studied below.
	\end{remark}
	
	\section{Numerical Results}\label{Numerical}
	In this section, we present and evaluate the performance of the achievable sum SE of  large SIM-assisted multiuser communications with statistical CSI by showing both analytical results and Monte Carlo simulations. For the setup, we assume that the SIM is parallel to the $ x-y $ plane and centered  along the $ z-$axis at a height $ H_{\mathrm{BS}}=10~\mathrm{m} $. The spacing  between adjacent meta-atoms  is assumed to be $ \lambda/2 $, and  the size of each meta-atom  is $ \lambda/2 \times \lambda /2 $.  The thickness of the SIM is $ T_{\mathrm{SIM}}=5 \lambda $, while the spacing is $ d_{\mathrm{SIM}}= T_{\mathrm{SIM}}/L$.  Moreover, the locations of the  users  are randomly distributed  at a distance between $60\mathrm{m}$ and $80\mathrm{m}$.
	
	The distance between the $ \tilde{n}-$th meta-atom of the $ (l-1)-$st metasurface and the  $ {n}-$th meta-atom of the $ l-$st metasurface is given by $ d_{n,\tilde{n}}^{l}=\sqrt{d_{\mathrm{SIM}}^{2}+d_{n,\tilde{n}}^{2}} $, where
	\begin{align}
		\!\!	d_{n,\tilde{n}}\!=\!\frac{\lambda}{2}\sqrt{\lfloor | n-\tilde{n}|/N_{x}\rfloor^{2}\!+\![ \mathrm{mod}(|n-\tilde{n}|,N_{x})]^{2}}.
	\end{align}
	The transmission distance between the $ m $-th antenna and the $ \tilde{n} $-th meta-atom on the first metasurface layer is provided by \eqref{dist1}. Note that we have $ \cos x_{n,\tilde{n}}^{l}= d_{\mathrm{SIM}}/ d_{n,\tilde{n}}^{l}, \forall l \in \mathcal{L}$.
	
	\begin{figure*}
		\begin{align}
			{\small 	d_{\tilde{n},m}^{1}\!=\!\sqrt{\!d_{\mathrm{SIM}}^{2}\!+\!\Big[\!\Big(\!\mathrm{mod}(\tilde{n}\!-\!1, N_{x})\!-\!\frac{N_{x}\!-\!1}{2}\!\Big)\frac{\lambda}{2}\!-\!\Big(m\!-\!\frac{N_{t}\!+\!1}{2}\Big)\frac{\lambda}{2}\Big]^{2}\!+\!\Big(\!\lceil \tilde{n}/N_{x} \rceil\!-\!\frac{N_{y}\!+\!1}{2}\Big)^{2}\frac{\lambda_{2}}{4}}}.\label{dist1}
		\end{align}
		\hrulefill
	\end{figure*}
	
	The  path loss   is given by 
	\begin{align}
		\tilde \beta_k = C_{0} (d_k/\hat{d})^{-\alpha},
	\end{align} 
	where $ C_{0} =(\lambda_{2}/4 \pi \hat{d})$ is  the free space path loss at the reference distance of $ \hat{d}=1~ \mathrm{m}$, and $\alpha=2.5$ is the path-loss exponent.  The correlation matrix $ \bR_{\mathrm{SIM}}$ is obtained according to \eqref{corC}.  The  carrier frequency and the system bandwidth are $ 2~\mathrm{GHz} $ and $ 20~\mathrm{MHz} $, respectively.  Moreover, we assume $ N_{t}=8 $, $ K=8 $, $ N=200 $, and  $ L=4 $.
	
	In Fig. \ref{fig2}, we depict the achievable sum SE versus the number of   elements $ N $ of each surface while varying the number of surfaces $L$. First, it is shown that  the downlink sum SE increases with $ N $ for different $L$.   Moreover, an increase in the number of surfaces results in an increase in the sum SE.  In addition,  for the sake of comparison, we present the performance  in the case of instantaneous CSI for $L=4$ \textcolor{black}{\cite{An2023c}}, which performs better than the case of statistical CSI since the latter is obtained based on a lower bound that is optimized at every several coherence intervals instead of at each coherence interval. However, the statistical CSI modeling allows to  save significant overhead. \textcolor{black}{Moreover, we show the
		effect of the size of each surface element. We observe that  as the size of each 			surface element decreases, the correlation decreases, and the sum 		SE increases. } Notably, Monte Carlo (MC) simulations corroborate the analytical results.

	\begin{figure}%
		\centering
		\includegraphics[width=0.8\linewidth]{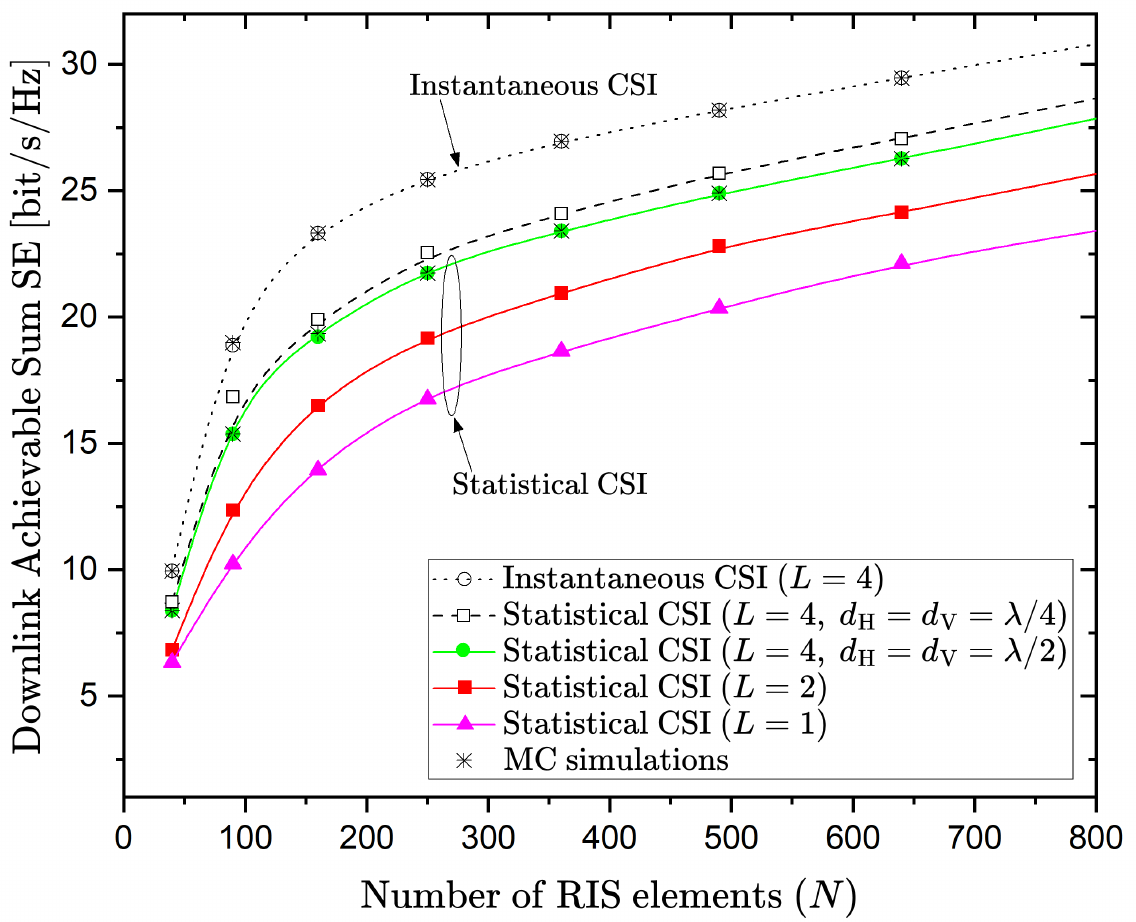}
		\caption{Achievable sum SE of the large SIM-aided MIMO architecture with respect to the number of meta-atoms  $ N $.}
		\label{fig2}
	\end{figure}
	
	Fig. \ref{fig3}  shows the sum SE versus the number of layers $ L $  of the SIM for the number of UEs $K$. When $N_{t}=K=8$, we observe that the sum SE improves until $ L=6 $  because the  SIM is able to mitigate the inter-user interference in the EM wave domain.   In particular, a significant improvement is observed compared to the single-layer SIM. Again, we illustrate the comparison between the cases corresponding to instantaneous and statistical CSI, where the latter exhibits worse performance for the benefit of lower overhead.

	\begin{figure}%
		\centering
		\includegraphics[width=0.8\linewidth]{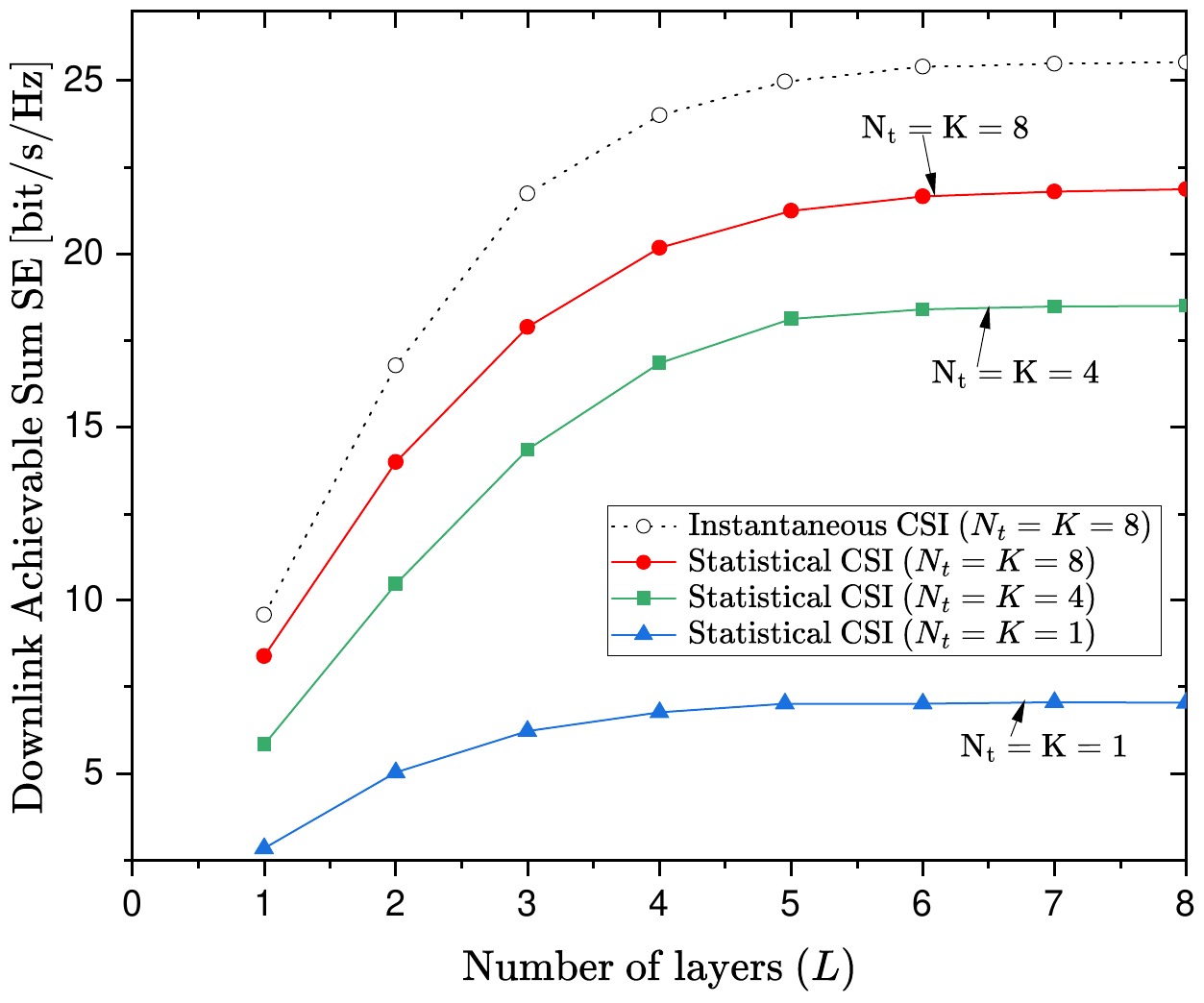}
		\caption{\textcolor{black}{Achievable sum SE of the large SIM-aided MIMO architecture with respect to the number of meta-surfaces  $ L $.}}
		\label{fig3}
	\end{figure}
	
	In Fig. \ref{fig4}, we depict the convergence of the proposed algorithm. The algorithm terminates  when the difference of the objective between the two last iterations is less than $10^{-5} $ or the number of iterations is larger than $ 130 $. It is shown that the algorithm converges to its maximum as the number of iterations increases.  Since Problem $ (\mathcal{P}) $ is non-convex, the algorithm does not converge to a globally optimal solution. This means that the algorithm  may converge to different points starting from different initial points. For this reason, we select the best   solutions after executing the algorithm from different initial points. Herein,  we depict the sum SE versus  the iteration count  for $ 5 $ different randomly generated initial points, and we observe that all points result in the same SE. Generally, the selection of $ 5 $ randomly generated initial points allows a good trade-off between performance and complexity.
	
	\begin{figure}%
		\centering
		\includegraphics[width=0.8\linewidth]{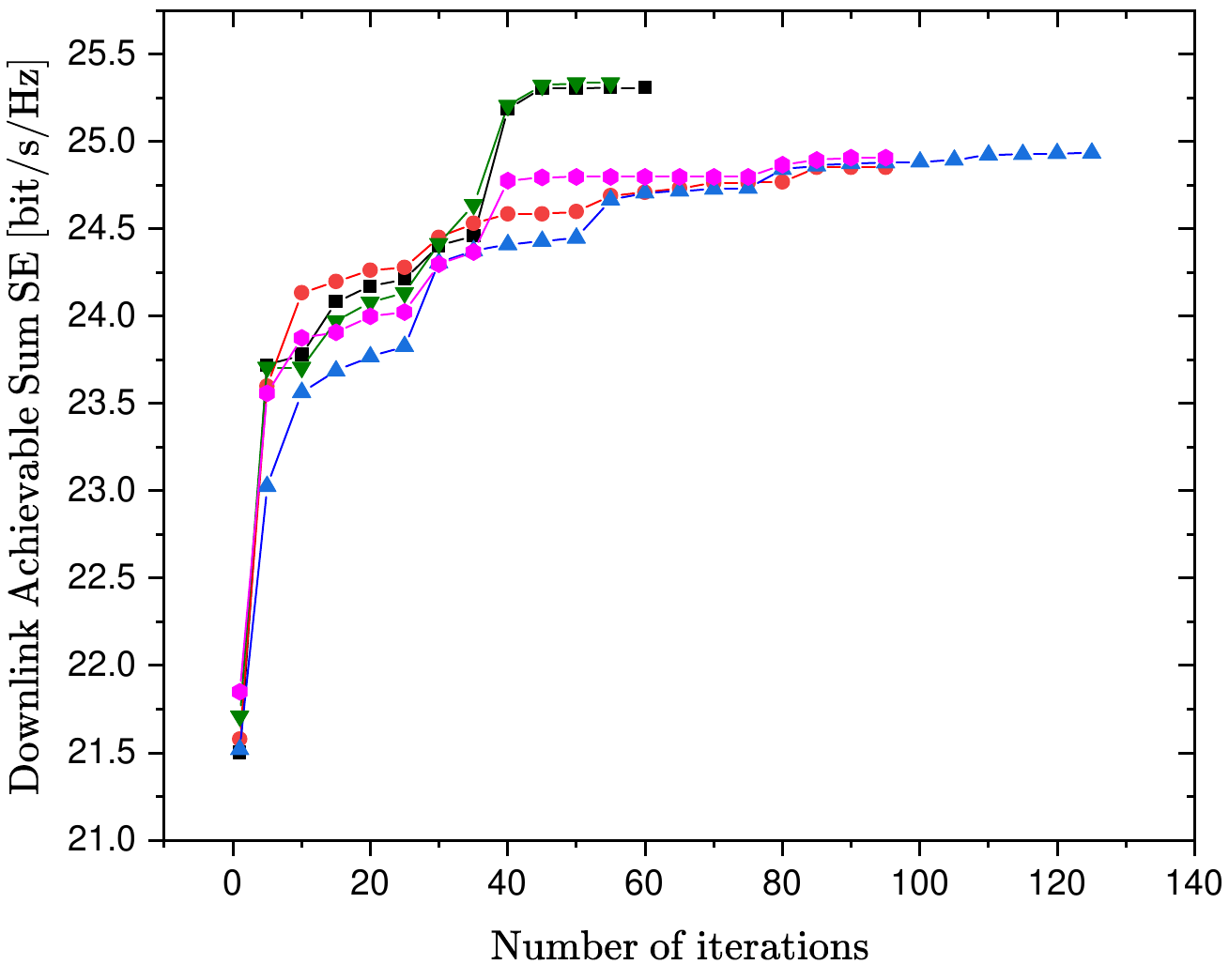}
		\caption{Achievable sum SE of the large SIM-aided MIMO architecture with respect to the number of iterations for   $ 5 $ different randomly generated initial points.}
		\label{fig4}
	\end{figure}

	\section{Conclusion} \label{Conclusion} 
	
	This paper provided  the achievable downlink SE of large SIM-aided multiuser communications over Ricean fading channels. In particular, we deduced a new, tractable expression for the downlink SE for large-metasurfaces as a function of  large-scale statistics while the downlink precoding takes place in the wave domain, in order to lower the computational burden and the processing latency. The  results were used to pursue an optimization based on statistical CSI that achieves lower overhead with respect to instantaneous CSI. Specifically, we proposed an AO algorithm that solves the optimization regarding the phase shifts of each surface of the SIM and the allocated power, which contributes to a reduction in processing latency due to lower overhead.

	\begin{appendices}
		\section{Proof of Proposition~\ref{propositionGradient}}\label{proposition1}	
		The proof  starts with the derivation of $ 	\nabla_{\bphi_{l}}D_k  $. To this end, we focus on the differential of $ |\bh_{k,\mathrm{LoS}}^{\H}\bG \bw_{k}^{1}|^{2} $. We have
		\begin{align}
			&	d(|\bh_{k,\mathrm{LoS}}^{\H}\bG \bw_{k}^{1}|^{2})=\bh_{k,\mathrm{LoS}}^{\H}d(\bG) \bw_{k}^{1}\bw_{k}^{1^{\H}}\bG^{\H}\bh_{k,\mathrm{LoS}}\nn\\
			&+	\bh_{k,\mathrm{LoS}}^{\H}\bG \bw_{k}^{1}\bw_{k}^{1^{\H}}d(\bG^{\H})\bh_{k,\mathrm{LoS}}\\
			&=\tr\big(\bh_{k,\mathrm{LoS}}^{\H}\bA_{l} d(\bPhi_{l})\bC_{l} \bw_{k}^{1}\bw_{k}^{1^{\H}}\bG^{\H}\bh_{k,\mathrm{LoS}}\big)\nn\\
			&+\tr(\bA_{l}^{\H}\bh_{k,\mathrm{LoS}}\bh_{k,\mathrm{LoS}}^{\H}\bG \bw_{k}^{1}\bw_{k}^{1^{\H}}\bC_{l}^{\H}d(\bPhi_{l}^{\H}))\label{num2}.	
		\end{align}
		In \eqref{num2}, we have substituted $d(\bG)=\bA_{l} d(\bPhi_{l})\bC_{l}$, where $ 	\bA_{l}=\bPhi_{L}\bW^{L}\cdots\bPhi_{l+1}\bW^{l+1} $, and $\bC_{l}= \bW^{l}\bPhi_{l-1}\bW^{l-1}\cdots \bPhi_{1} $. Having obtained the differential, we have
		\begin{align}
			\nabla_{\bphi_{l}}D_k&=\frac{\partial}{\partial\bphi_{l}^{*}}D_k\\
			&=p_{k}	\kappa_{k}\diag(\bC_{l}^{*}\bw_{k}^{1^{*}}\bw_{k}^{1^{\T}}\bG^{\T}	\bh_{k,\mathrm{LoS}}^{*}\bh_{k,\mathrm{LoS}}^{\T}\bA_{l}^{*}).
		\end{align}
		
		The first term in the denominator is written as
		\begin{align}
			&d\big(\tr\big(\bG \bw_{i}^{1}\bw_{i}^{1^{\H}}\bG^{\H}\bR\big)\big)=\tr\big(\bA_{l} d(\bPhi_{l})\bC_{l} \bw_{i}^{1}\bw_{i}^{1^{\H}}\bG^{\H}\bR\big)\nn\\
			&+\tr\big(\bA_{l}^{\H}\bR\bG \bw_{i}^{1}\bw_{i}^{1^{\H}}\bC_{l}^{\H}d(\bPhi_{l}^{\H})\big).
		\end{align}
		
		Thus, we have
		\begin{align}
			\!	\nabla_{\bphi_{l}}\!\tr\big(\bG \bw_{i}^{1}\bw_{i}^{1^{\H}}\bG^{\H}\bR\big)\!=\!\diag(\bC_{l}^{*}\bw_{i}^{1^{*}}\bw_{i}^{1^{\T}}\bG^{\T}	\bR\bA_{l}^{*}).
		\end{align}
		The second term in the denominator is similar to the numerator. Hence, the derivation is similar.
	\end{appendices}
	\bibliographystyle{IEEEtran}
	
	\bibliography{IEEEabrv,bibl}
\end{document}